\documentclass{article}

\usepackage{algorithm2e}
\usepackage{tikz}
\usetikzlibrary{matrix,arrows,decorations.pathmorphing}

\usepackage{amsmath}
\usepackage{amsthm}
\usepackage{amsfonts}
\usepackage{amssymb}

\usepackage{url}

\newtheorem{lemma}{Lemma}
\newtheorem{theorem}{Theorem}

\newtheorem{proposition}[theorem]{Proposition}
\newtheorem{cor}[theorem]{Corollary}

\theoremstyle{definition}

\newtheorem{krule}{Reduction Rule}

\theoremstyle{remark}
\newtheorem{rem}{Remark}

\newcommand{\2}{\vspace{0.15cm}}

    \newenvironment{myindentpar}[1]%
     {\begin{list}{}%
             {\setlength{\leftmargin}{#1}}%
             \item[]%
     }
     {\end{list}}

\title{Parameterizations of Test Cover with Bounded Test Sizes}
\author{R. Crowston, G. Gutin, M. Jones and G. Muciaccia\\
{\small Royal Holloway, University of London}\\[-3pt]
{\small Egham, Surrey, TW20 0EX, UK}\\[-3pt]
{\small \url{{robert|gutin|markj|G.Muciaccia}@cs.rhul.ac.uk}}\\
\and A. Yeo\\
{\small University of Johannesburg}\\[-3pt]
{\small Auckland Park, 2006 South Africa}\\[-3pt]
{\small \url{andersyeo@gmail.com}}\\
}

\begin{document}
\maketitle

\begin{abstract}
In the {\sc Test Cover} problem we are given a hypergraph $H=(V, \mathcal{E})$ with $|V|=n, |\mathcal{E}|=m$, and we assume that $\mathcal{E}$ is a \emph{test cover}, i.e. for every pair of vertices $x_i, x_j$, there exists an edge $e \in \mathcal{E}$ such that $|\{x_i,x_j\}\cap e|=1$. The objective is to find a minimum subset of $\mathcal{E}$ which is a test cover. The problem is used for identification across many areas, and is NP-complete. From a parameterized complexity standpoint, many natural parameterizations of {\sc Test Cover} are either $W[1]$-complete or have no polynomial kernel unless $coNP\subseteq NP/poly$, and thus are unlikely to be solveable efficiently.

However, in practice the size of the edges is often bounded. In this paper we study the parameterized complexity of {\sc Test-$r$-Cover}, the restriction of {\sc Test Cover} in which each edge contains at most $r \ge 2$ vertices. In contrast to the unbounded case, we show that the following below-bound parameterizations of {\sc Test-$r$-Cover} are fixed-parameter tractable with a polynomial kernel: (1) Decide whether there exists a test cover of size $n-k$, and (2) decide whether there exists a test cover of size $m-k$, where $k$ is the parameter.
In addition, we prove a new lower bound $\lceil \frac{2(n-1)}{r+1} \rceil$ on the minimum size of a test cover when the size of each edge is bounded by $r$. {\sc Test-$r$-Cover} parameterized above this bound is unlikely to be fixed-parameter tractable; in fact, we show that it is para-NP-complete, as it is NP-hard to decide whether an instance of {\sc Test-$r$-Cover} has a test cover of size exactly  $\frac{2(n-1)}{r+1}$.
\end{abstract}

\section{Introduction}

The input to the {\sc Test Cover} problem consists of a hypergraph $H=(V,\mathcal{E})$,
with vertex set $V=\{x_1,\dots,x_n\}$ and edge set $\mathcal{E}=\{e_1,\dots,e_m\}$
\footnote{Notice that, in the literature, the vertices are also called
{\em items} and the edges {\em tests}. We provide basic terminology and notation on hypergraphs in Section \ref{sec:prelim}.}.
We say that
an edge $e_q$ {\em separates} a pair of vertices $x_i,x_j$ if $|\{x_i,x_j\}\cap e_q|=1$. A subcollection ${\cal T}\subseteq {\cal E}$
is a {\em test cover} if each pair of distinct vertices $x_i,x_j$ is separated by an edge in ${\cal T}$.
The objective is to find a test cover of minimum size, if one exists. Since it is easy to decide,
in polynomial time, whether the collection ${\cal E}$ itself is a test cover,
henceforth we will assume that ${\cal E}$ is a test cover.

{\sc Test Cover} arises naturally in the following general setting of
identification problems: Given a set of items (which corresponds to the set of vertices) and a set
of binary attributes that may or may not occur in each item, the aim
is to find the minimum size subset of attributes (corresponding to a minimum test cover) such that each
item can be uniquely identified from
the list of attributes it has from this subset.
{\sc Test Cover} arises in fault analysis, medical diagnostics, pattern recognition and
biological  identification (see, e.g., \cite{HalHalRav01,HalMinRav01,MorSha85}).

The {\sc Test Cover} problem has also been studied extensively from an algorithmic viewpoint.
The problem is NP-hard, as was shown by Garey and Johnson \cite{GareyJohnson}.
There is an $O(\log n)$-approximation algorithm for
the problem \cite{MorSha85} and there is no $o(\log n)$-approximation algorithm unless P=NP \cite{HalHalRav01}.
Often, in practice, the sizes of all edges are bounded from above by a constant $r \ge 2$ \cite{De,HalHalRav01}.
For such cases the problem remains NP-hard, even when $r=2$, but the approximation guarantee can be improved to $O(\log r)$ \cite{HalHalRav01}.

Initially, research in parameterized algorithmics\footnote{We provide basic notions on parameterized algorithmics in the end of this section.} considered mainly
standard parameterizations, where the parameter is the size of the
solution. In the last decade, this situation has changed: many papers deal
with structural parameters such as treewidth or parameterizations above and
below tight bounds. Research on parameterizations above and below tight
bounds was initiated by Mahajan and Raman \cite{MR99}, but the systematic
study of such parameterizations started only after the publication of
\cite{MahajanRamanSikdar09}. This resulted in new methods and
approaches, solving most of the open problems stated in \cite{MR99} and
\cite{MahajanRamanSikdar09}, see, e.g.,
\cite{AloGutKimSzeYeo11,CFGJRTY,CroGutJonYeo12,CroJonMni,GutKimMniYeo}.

The complexity of the following four parameterizations of {\sc Test Cover} were first studied in \cite{CrowstonGJSY12,GutinMY12}, in all of which $k$ is
the parameter:
\begin{description}
  \item[{\sc TestCover}($k$):] Is there a test cover with at most $k$ edges?
  \item[{\sc TestCover}($n-k$):] Is there a test cover with at most $n-k$ edges?
  \item[{\sc TestCover}($m-k$):] Is there a test cover with at most $m-k$ edges?
  \item[{\sc TestCover}($\log n + k$):] Is there a test cover with at most $\log n + k$ edges?
\end{description}
Whilst {\sc TestCover}($k$) is a standard parameterization, {\sc TestCover}($n-k$) and {\sc TestCover}($m-k$) are parameterizations below
tight upper bounds: clearly $m$ is a tight upper bound on the minimum size of a test cover and it is not hard to prove that $n-1$ is another
tight upper bound \cite{Bon72} (see also Corollary \ref{cor:bondy}). Finally, {\sc TestCover}($\log n + k$) is a parameterization above a tight lower bound
$\lceil \log n\rceil$ on the minimum size of a test cover \cite{HalHalRav01}.

It is clear that {\sc TestCover}($k$) is fixed-parameter tractable; it follows immediately from the tight lower bound $\lceil \log n\rceil$.
It is proved in \cite{CrowstonGJSY12} that {\sc TestCover}($n-k$) is fixed-parameter tractable. The proof of this result is quite nontrivial
and yields an algorithm whose runtime is $f(k)(n+m)^{O(1)}$, where $f(k)$ is a function growing impractically fast.
It is also proved in \cite{CrowstonGJSY12} that {\sc TestCover}($m-k$) and {\sc TestCover}($\log n + k$) are W[1]-hard.
Since {\sc TestCover}($k$) and {\sc TestCover}($n-k$) are fixed-parameter tractable, it is natural to ask whether they admit polynomial-size
kernels. The authors of \cite{GutinMY12} answered this question by proving that, unless $coNP\subseteq NP/poly$,
neither {\sc TestCover}($k$) nor {\sc TestCover}($n-k$) admits a polynomial-size kernel.

Thus, papers \cite{CrowstonGJSY12,GutinMY12} demonstrated that the four parameterizations of {\sc Test Cover} are not easy to solve:
even those of them which are fixed-parameter tractable do not admit polynomial-size kernels. This gives another theoretical explanation of
the fact that {\sc Test Cover} is not easy to solve in practice \cite{Debnb,FahTie}. Thus, it is natural to study important special cases
of {\sc Test Cover}. It turns out that often the sizes of all edges are bounded from above by a constant $r$: De Bontridder
{\em et al.} \cite{De} observed that ``this is the common restriction" for the problem and provide, as an example, protein identification.
A question is whether the parameterizations of {\sc Test Cover} become easier in this case. Already \cite{GutinMY12} indicated that this can
be true by proving that {\sc TestCover}($k$) does admit a polynomial-size kernel if $r$ is a constant. We will denote the special cases of the
four parameterizations by {\sc Test-$r$-Cover}($k$),
{\sc Test-$r$-Cover}($n-k$), {\sc Test-$r$-Cover}($m-k$), and {\sc Test-$r$-Cover}($\log n + k$).
Since the case $r=1$ is trivial, henceforth we will assume $r \ge 2$.

In this paper we will prove that not only is {\sc Test-$r$-Cover}($m-k$) fixed-parameter tractable, but {\sc Test-$r$-Cover}($m-k$)
and {\sc Test-$r$-Cover}($n-k$) also  admit polynomial-size kernels. There is no interest in studying {\sc Test-$r$-Cover}($\log n + k$) as
$\lceil \log n\rceil$ is no longer a tight lower bound. 
Instead, we prove that $\lceil \frac{2(n-1)}{r+1} \rceil$ is a tight lower bound and study the parameterization {\sc Test-$r$-Cover}($\frac{2(n-1)}{r+1} + k$): Is there a test cover with at most $\frac{2(n-1)}{r+1} + k$ edges, where $k$ is the parameter? We prove that it is NP-hard even to decide whether there is a test cover of size $\frac{2(n-1)}{r+1}$. It follows that {\sc Test-$r$-Cover}($\frac{2(n-1)}{r+1} + k$) is para-NP-complete. 
 Note that lower bounds play a key
role in designing branch-and-bound algorithms for {\sc Test Cover} \cite{Debnb,FahTie} and so
our new lower bound may be of interest for practical solutions of {\sc Test-$r$-Cover}.

The rest of the paper is organized as follows. In the remainder of this section we provide necessary basic terminology on parameterized algorithmics. In Section \ref{sec:prelim}, we provide necessary terminology on hypergraphs and
prove some preliminary results. In Section \ref{sec:ALB}, we show that $\lceil \frac{2(n-1)}{r+1} \rceil$ is a tight lower bound and that
{\sc Test-$r$-Cover}($\frac{2(n-1)}{r+1}+k$) is para-NP-complete.
In Section \ref{sec:m-kfpt}, we give a short proof that {\sc Test-$r$-Cover}($m-k$) is fixed-parameter tractable.
In Sections \ref{sec:m-k} and \ref{sec:n-k}, we prove that {\sc Test-$r$-Cover}($m-k$) and {\sc Test-$r$-Cover}($n-k$) admit
polynomial-size kernels.
In Section \ref{sec:d}, we discuss open problems.

{\bf Basics on Parameterized Complexity.} A parameterized problem $\Pi$ can be considered as a set of pairs
$(I,k)$ where $I$ is the \emph{problem instance} and $k$ (usually a nonnegative
integer) is the \emph{parameter}.  $\Pi$ is called
\emph{fixed-parameter tractable (fpt)} if membership of $(I,k)$ in
$\Pi$ can be decided by an algorithm of runtime $O(f(k)|I|^c)$, where $|I|$ is the size
of $I$, $f(k)$ is an arbitrary function of the
parameter $k$ only, and $c$ is a constant
independent from $k$ and $I$. 
Let $\Pi$ and $\Pi'$ be parameterized
problems with parameters $k$ and $k'$, respectively. An
\emph{fpt-reduction $R$ from $\Pi$ to $\Pi'$} is a many-to-one
transformation from $\Pi$ to $\Pi'$
that maps an instance $(I,k)$ of $\Pi$ to an instance $(I',k')$ of $\Pi'$, such that (i) $(I,k)\in \Pi$ if
and only if $(I',k')\in \Pi'$ with $k'\le g(k)$ for a fixed
function $g$, and (ii) $R$ is of complexity
$O(f(k)|I|^c)$.


$\Pi$ is in \emph{para-NP} if membership of $(I,k)$ in
$\Pi$ can be decided by a nondeterministic Turing machine in time $O(f(k)|I|^c)$,
where $|I|$ is the size of $I$, $f(k)$ is an arbitrary function of the
parameter $k$ only,
and $c$ is a constant independent from $k$ and $I$. A parameterized problem $\Pi'$ is {\em
para-NP-complete} if it is in para-NP and for any parameterized
problem $\Pi$ in para-NP there is an fpt-reduction from $\Pi$ to
$\Pi'$. It is well-known that a
parameterized problem $\Pi$ belonging to para-NP is para-NP-complete if we can reduce an
NP-complete problem to the subproblem of $\Pi$ when the parameter is equal to some
constant \cite{FlumGrohe06}.

Given a parameterized problem $\Pi$,
a \emph{kernelization of $\Pi$} is a polynomial-time
algorithm that maps an instance $(I,k)$ to an instance $(I',k')$ (the
\emph{kernel}) such that (i)~$(I,k)\in \Pi$ if and only if
$(I',k')\in \Pi$, and (ii)~ $|I'|+k'\leq g(k)$ for some
function $g$ of $k$ only. We call $g(k)$ the {\em size} of the kernel.
It is well-known \cite{DowneyFellows99,FlumGrohe06} that a decidable parameterized problem $\Pi$ is fixed-parameter
tractable if and only if it has a kernel. Polynomial-size kernels are of
main interest, due to applications \cite{DowneyFellows99,FlumGrohe06,Niedermeier06}, but unfortunately not all fixed-parameter problems
have such kernels unless  coNP$\subseteq$NP/poly, see, e.g., \cite{BDFH09,DLS09}. 


\section{Additional terminology, notation and preliminaries}\label{sec:prelim}

%

We use usual hypergraph terminology. Let $H=(V,\mathcal{E})$ be a
hypergraph. The degree of a vertex $x$ is the cardinality of
$\{e\in\mathcal{E}:x\in e\}$. Given $X\subseteq V$, define the neighborhood
$N_1(X)=\{u\in V\setminus X:\exists e\in\mathcal{E} \ \exists x\in X \text{ such that } \{x,u\} \subseteq e \}$,
$N_1[X]=N_1(X)\cup X$ and $N_j[X]=N_1[N_{j-1}[X]]$.
Given $\mathcal{F}\subseteq\mathcal{E}$, define the neighborhood
$N_1({\cal F}) = \{e \in
 {\cal E} \setminus {\cal F} \; : \; \exists f \in {\cal F}, \; \; f \cap e
 \not= \emptyset \}$,
 $N_1[\mathcal{F}]=N_1(\mathcal{F})\cup \mathcal{F}$ and $N_j[\mathcal{F}]=N_1[N_{j-1}[\mathcal{F}]]$.
 For $\mathcal{F}\subseteq\mathcal{{E}}$, $H[\mathcal{F}]$ is the hypergraph with vertex set $V(H[\mathcal{F}])=\{x\in V:\exists e\in\mathcal{F} \text{ such that } x\in e\}$ and edge set $\mathcal{E}(H[\mathcal{F}])=\mathcal{F}$.
 Similarly, given a set of vertices $X\subseteq V$, $H[X]$ is the hypergraph with vertex set $V(H[X])=X$ and edge set $\mathcal{E}(H[X])=\{e\cap X : e \in\mathcal{E}, e\cap X\neq \emptyset \}$.

%
%
%

We reuse some terminology from \cite{CrowstonGJSY12}.
For a subset $\mathcal{T} \subseteq \mathcal{E}$, the \emph{classes induced by $\mathcal{T}$} are the maximal subsets $C_i$ of $V$ such that no pair of vertices in $C_i$ is separated by an edge in $\mathcal{T}$. Observe that the classes induced by $\mathcal{T}$ form a partition of $V$.

We say that $\mathcal{E}$ \emph{separates} $X\subseteq V$ and $Y\subseteq V$, $X \cap Y = \emptyset$, if for every pair $(x,y)$, with $x\in X$ and $y\in Y$, $x$ and $y$
are separated by an edge of $\mathcal{E}$. We say that $\mathcal{E}$ \emph{isolates} $X\subseteq V$ if it separates $X$ and $V\setminus X$.
We say an edge $e$ \emph{cuts} $X \subseteq V$ if $X \cap e \neq \emptyset$ and $X \setminus e \neq \emptyset$.

For an integer $t$, we use $[t]$ to denote the set $\{1,2, \dots, t\}$.

\begin{lemma}\label{lem:classes}
If $\cal E$ induces $t \ge 2$ classes in a hypergraph $H=(V,{\cal E})$ and $i\in
[t-1]$ then there is a subset $\cal F$ of $\cal E$ with $i$ edges that
induces at least $i+1$ classes.
\end{lemma}
\begin{proof}
By induction on $i\in [t-1]$. To see that the lemma holds for $i=1$ set
${\cal F}=\{e\}$, where $e$ is any edge of $\cal E$ with less than $|V|$
vertices. Let $\cal F$ be a subset of $\cal E$ with $i-1$ edges that
induces at least $i$ classes, 
let $x,y$ be vertices separated by $\cal E$
not separated by ${\cal F}$, and let $e$ be an edge separating $x$ and $y$.
 It remains to
observe that ${\cal F}\cup \{e\}$ induces at least $i+1$ classes.
\end{proof}

Observe that $\cal E$ is a test cover if and only if it induces $n$
classes. This and the lemma above imply the following:

\begin{cor}\label{cor:bondy}\cite{Bon72}
If $\cal E$ is a test cover in a hypergraph $H=(V,{\cal E})$ then there is
a subset $\cal F$ of $\cal E$ with at most $n-1$ edges which is also a test
cover.
\end{cor}

\begin{cor}\label{cor:sep}
In a hypergraph $H=(V,{\cal E})$, let $X, Y\subseteq V$ be such that $X\cap Y=\emptyset$ and ${\cal E}$ separates 
$X$ and $Y$.
Then there is a subset $\cal F$
of $\cal E$ that separates $X$ and $Y$ and has at most $t_X + t_Y -1$ edges, where $t_X$ ($t_Y$) is the number of classes induced by ${\cal E}$  that intersect $X$ ($Y$).
\end{cor}
\begin{proof}
Apply Lemma \ref{lem:classes} to $H[X\cup Y]$ and then extend the obtained
edges of $H[X\cup Y]$ beyond $X\cup Y$ such that they correspond to edges
of $H.$
\end{proof}

\section{{\sc Test-$r$-Cover$(\frac{2(n-1)}{r+1}+k)$}}\label{sec:ALB}

\begin{proposition}\label{prop:bound}
If the size of every edge is at most $r$, then every test cover has at least $\lceil \frac{2(n-1)}{r+1} \rceil$ edges. This lower bound on the size of a test cover is tight.
Furthermore, for a test cover of size exactly the lower bound, every vertex is contained in at most two edges.
\end{proposition}
\begin{proof}
We first prove that $\lceil \frac{2(n-1)}{r+1} \rceil$ is indeed a lower bound. 
Given a test cover of size $m'$, observe that at most one vertex may be contained in no edges. For each of the $m'$ edges, at most one vertex is contained in only that edge and every other vertex is contained in at least two edges. Hence 
$n \le 1 + m' + \frac{m'(r-1)}{2}$,
 which implies $m' \ge \frac{2(n-1)}{r+1}$. Observe that no vertex being contained in three edges is a necessary condition for the bound to be tight.

To see that this bound is tight, consider a set $V=\{x_{i,j}: i,j\in [r]\}$ of vertices, and a set
${\cal E}=\{e_q:\ q\in [r-1]\}\cup \{e'_s:\ s\in [r-1]\}$ of edges, where
$e_q=\{x_{q,j} :\ j\in [r]\}$, $e'_s=\{ x_{i,s} : i\in [r]\}$
(see Figure \ref{fig:prop1}).
 Since $n=r^2$, we have $|{\cal E}|=2(r-1)=\frac{2(n-1)}{r+1}.$
Consider two vertices $x_{i,j}$ and $x_{i',j'}$. If $i\neq i'$, then $x_{i,j}$ and $x_{i',j'}$ are separated
by $e_{\min\{i,i'\}}$ and if $j\neq j'$, then $x_{i,j}$ and $x_{i',j'}$ are separated
by $e'_{\min\{j,j'\}}$. Thus, ${\cal E}$ is a test cover of minimum possible size.
By using multiple copies of this construction (except for $x_{r,r}$), one can see that the bound is tight for arbitrarily large $n$.
\end{proof}

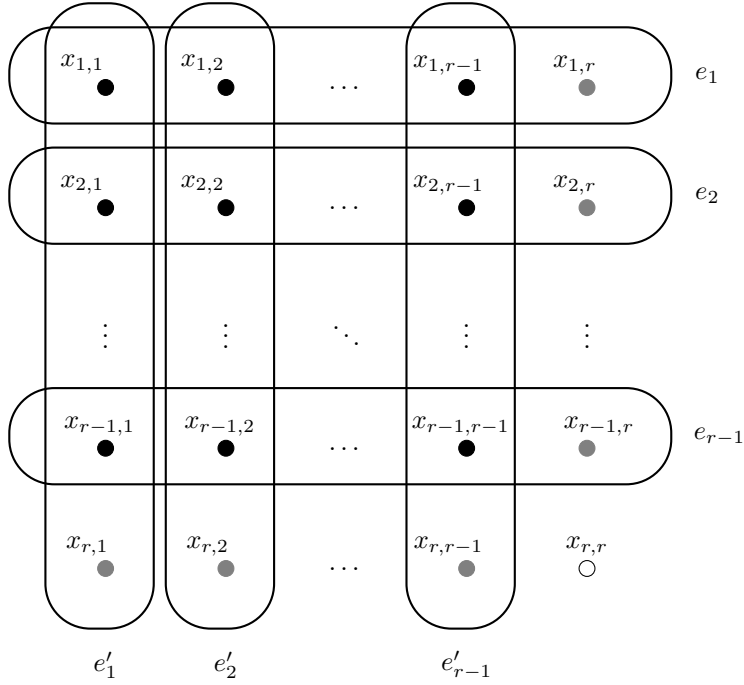
\begin{figure}[h]
\centering
\begin{tikzpicture}[scale=1.6]

\foreach \number in {0,-1,-3,-4}
{
\node at (2,\number) {$\dots$};
}

\foreach \number in {0,1,3,4}
{
\node at (\number, -2) {$\vdots$};
}

\node at (2,-2) {$\ddots$};

\node at (2.85,0.185) {$x_{1,r-1}$};
\node at (2.85,0.185-1) {$x_{2,r-1}$};

\node at (-0.05,0.185-3) {$x_{r-1,1}$};
\node at (0.95,0.185-3) {$x_{r-1,2}$};
\node at (2.95,0.185-3) {$x_{r-1,r-1}$};

\node at (3.9,0.185) {$x_{1,r}$};
\node at (3.9,0.185-1) {$x_{2,r}$};
\node at (4.1,0.185-3) {$x_{r-1,r}$};

\node at (-0.15,0.185-4) {$x_{r,1}$};
\node at (0.85,0.185-4) {$x_{r,2}$};
\node at (2.85,0.185-4) {$x_{r,r-1}$};

\node at (4,0.185-4) {$x_{r,r}$};

\node at (5,0.1) {$e_1$};
\node at (5,-0.9) {$e_2$};
\node at (5.1,-2.9) {$e_{r-1}$};

\node at (0,-4.8) {$e_{1}'$};
\node at (1,-4.8) {$e_{2}'$};
\node at (3,-4.8) {$e_{r-1}'$};

\tikzstyle{every node}=[draw,circle,fill=black,minimum size=6pt,inner sep=0pt]

\foreach \numberone in {1,2}
{
\foreach \numbertwo in {1,2}
{
\node at (\numberone-1,1-\numbertwo) [label=above left:$x_{\numbertwo,\numberone}$,font=\footnotesize] {};
}
}

\foreach \number in {1,2}
{
\node at (3,1-\number)  {};
}

\foreach \number in {0,1,3}
{
\node at (\number,-3)  {};
}

\tikzstyle{every node}=[draw,gray,circle,fill=gray,minimum size=6pt,inner sep=0pt]

\foreach \number in {0,-1,-3}
{
\node at (4,\number)  {};
}

\foreach \number in {0,1,3}
{
\node at (\number,-4)  {};
}

\tikzstyle{every node}=[draw,circle,fill=white,minimum size=6pt,inner sep=0pt]

\node at (4,-4) {};

\draw [thick,rounded corners=17pt] (-0.8,0.5) rectangle (4.7,-0.3);
\draw [thick,rounded corners=17pt] (-0.8,-0.5) rectangle (4.7,-1.3);
\draw [thick,rounded corners=17pt] (-0.8,-2.5) rectangle (4.7,-3.3);

\draw [thick,rounded corners=17pt] (-0.5,0.7) rectangle (0.4,-4.5);
\draw [thick,rounded corners=17pt] (0.5,0.7) rectangle (1.4,-4.5);
\draw [thick,rounded corners=17pt] (2.5,0.7) rectangle (3.4,-4.5);

\end{tikzpicture}
\caption{Illustration of the hypergraph $(V, \mathcal{E})$ in Proposition \ref{prop:bound}.
 Vertices of degree two are labelled in black, vertices of degree one in gray
and the vertex of degree zero in white.}\label{fig:prop1}
\end{figure}

From this lower bound, a kernel for {\sc Test-$r$-Cover}($k$), with less vertices than in \cite{GutinMY12}, immediately follows:

\begin{cor}
{\sc Test-$r$-Cover}$(k)$ admits a kernel with at most $k(r+1)/2+1$ vertices.
\end{cor}

It follows from the proof of Proposition \ref{prop:bound} that for a test cover of size $\frac{2(n-1)}{r+1}$, there must be exactly one vertex contained in no edges, each edge must contain one vertex of degree $1$, and every other vertex must be of degree $2$.

%

To prove 
the next lemma
 we give a reduction from
the {\sc $r$-Dimensional Matching} problem, which is one of Karp's 21 NP-complete problems \cite{Karp72}, for each $r\ge 3$.
To state the problem, we will give the following definitions. A collection $\cal E$ of edges of a hypergraph $H$ is called a {\em matching} if no two edges
of $\cal E$ have a common vertex. A matching is {\em perfect} if every vertex of $H$ belongs to an edge of the matching. A hypergraph $H$ is {\em $r$-partite $r$-uniform} if the vertex set $V(H)$ of $H$ can be partitioned into $r$ sets $X_1,\ldots ,X_r$ such that each edge of $H$ has exactly one vertex from each $X_i.$ In {\sc $r$-Dimensional Matching}, given an $r$-partite $r$-uniform hypergraph $H$, the aim is to decide whether $H$ has a perfect matching.

\begin{lemma} \label{lem:nphard}
{\sc Test-$r$-Cover}$(\frac{2(n-1)}{r+1}+k)$ is NP-hard for $k=0$ and $r\ge 3$.
\end{lemma}
\begin{proof}

We give a reduction from {\sc $r$-Dimensional Matching} to {\sc Test-$r$-Cover}$(\frac{2(n-1)}{r+1}+k)$ for $k=0$. We assume we have an instance of {\sc $r$-Dimensional Matching} given by an $r$-partite $r$-uniform hypergraph $G$ with vertex set $V(G)=\{x_{i,j} :\ i\in [r], j\in [n']\}$, where $X_i=\{x_{i,j} : j\in [n']\}$ form the partition of $V(G)$, i.e. each edge has exactly one vertex in each $X_i$. (We may assume all $X_i$ have the same size, as otherwise there is no perfect matching.) The edge set of $G$ will be denoted by $\mathcal{E}(G)$.
Additionally, we may assume that $n'$ is divisible by $r-1$.

We now give the construction of $(V,\mathcal{E})$, the {\sc Test-$r$-Cover} instance.
Let the vertex set $V=V(G)\cup Y$, and the edge set $\mathcal{E}=\mathcal{E}(G)\cup \mathcal{E}'$, where $Y$ and $\mathcal{E}'$ are defined as follows:
\[ Y=\{y_{i,j} :\ i\in [r-1], j\in \{(r-1),2(r-1),...,n'\} \} \cup \{ y_0 \}
\]
\[ \mathcal{E}'=\{e_{i,p} :\ i\in [r-1], p\in \{(r-1),2(r-1),...,n'\} \}
\]
where $e_{i,p}=\{x_{i,p-r+2}, x_{i,p-r+3}, \ldots, x_{i,p}, y_{i,p}\}$ (see Figure \ref{fig:lem3}).
Note that $n=|V|=rn'+n'+1$.

We first show that if $\mathcal{E}'' \subseteq\mathcal{E}(G)$ is a perfect matching in $G$, then $\mathcal{E}'\cup \mathcal{E}''$ is a test cover of size $\frac{2(n-1)}{r+1}$. 
For any pair $x_{i,j}, x_{i',j'} \in V(G)$, 
if $i=i'$ then  $x_{i,j}, x_{i',j'}$ appear in different edges in $\mathcal{E}''$ and are separated, if $i <i'$ they are separated by $e_{i,(r-1)\lceil \frac{j}{r-1}\rceil}$, and if $i>i'$ they are separated by $e_{i',(r-1)\lceil \frac{j'}{r-1}\rceil}$.
Any pair $y_{i,j}, y_{i',j'} \in Y$ is separated by $e_{i,j}$.
As $\mathcal{E}'$ covers $V(G)$, every $x_{i,j} \in V(G)$ is separated from every $y_{i',j'} \in Y$.
Finally, $y_0$ is separated from every other vertex as it is the only vertex not covered by $\mathcal{E}' \cup \mathcal{E}''$.
Thus, $\mathcal{E}' \cup \mathcal{E}''$ is a test cover.
Since $|\mathcal{E}' \cup \mathcal{E}''|=2n'$ and $n=rn'+n'+1$, we have that $\mathcal{E}' \cup \mathcal{E}''$ is a test cover of size $\frac{2(n-1)}{r+1}$.

It remains to show that if $(V,\mathcal{E})$ has a test cover of size $\frac{2(n-1)}{r+1}$, then $G$ has a perfect matching.
Firstly, observe that to separate $y_0$ from $y_{i,p}$, every edge $e_{i,p}$ must be in the test cover. Hence $\mathcal{E}'$ is contained in any test cover for $(V,\mathcal{E})$. So a test cover is a set $\mathcal{E}'\cup \mathcal{E}''$, where $\mathcal{E}'' \subseteq \mathcal{E}(G)$.
Next, observe that for every vertex $x_{i,j}$, there must exist an edge in $\mathcal{E}''$ containing $x_{i,j}$, as otherwise there would be no edge separating $x_{i,j}$ from $y_{i,(r-1)\lceil \frac{j}{r-1}\rceil}$.

We may also observe that no two edges in $\mathcal{E}''$ intersect. Suppose two edges intersected, but not in partite set $X_r$. Then there is another edge in $\mathcal{E}'$ also intersecting these edges at the same vertex. But we know from Proposition \ref{prop:bound} that in a test cover of size $\frac{2(n-1)}{r+1}$, at most two edges intersect at any vertex, so this cannot happen. Suppose instead, that edges $e$ and $e'$ intersect in partite set $X_r$. Then $e$ also intersects an edge in $\mathcal{E}'$ in each of the other partitions, so every vertex in $e$ is of degree $2$. But we know that in a test cover of size $\frac{2(n-1)}{r+1}$, one vertex in each edge has degree 1, so this case is also not possible. Hence, a test cover with $\frac{2(n-1)}{r+1}$ edges for $(V,\mathcal{E})$ would give a perfect matching in $G$.
\end{proof}

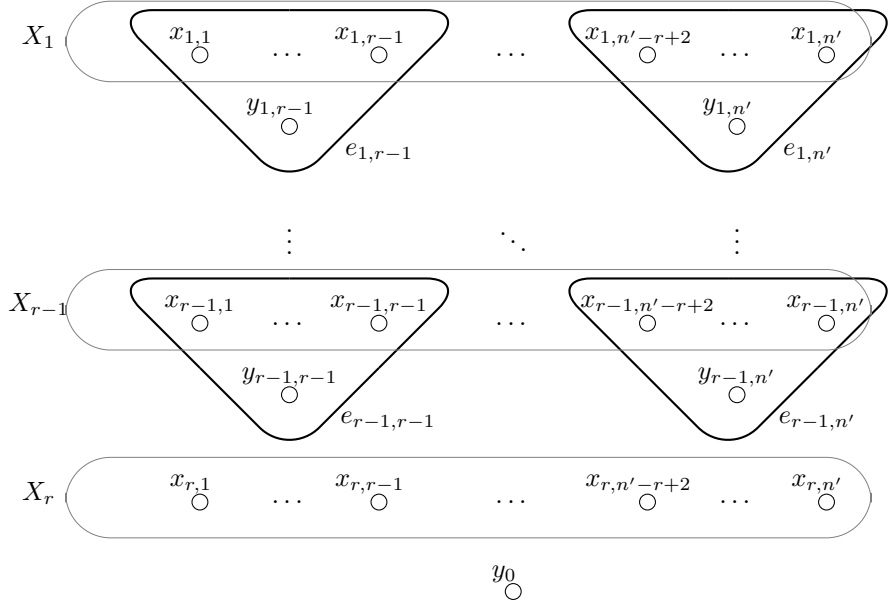
\begin{figure}[h]
\centering
\begin{tikzpicture}[scale=1.19]

\foreach \numbertwo in {0,-3,-5}
{
\foreach \numberone in {1,3.5,6}
{
\node at (\numberone,\numbertwo) {$\dots$};
}
}

\node at (1,-2) {$\vdots$}; \node at (3.5,-2) {$\ddots$}; \node at (6,-2) {$\vdots$};

\node at (-0.1,0.2) {$x_{1,1}$}; \node at (1.9,0.2) {$x_{1,r-1}$};
\node at (4.9,0.2) {$x_{1,n'-r+2}$}; \node at (6.9,0.2) {$x_{1,n'}$};

\node at (0,-2.8) {$x_{r-1,1}$}; \node at (2,-2.8) {$x_{r-1,r-1}$};
\node at (5,-2.8) {$x_{r-1,n'-r+2}$}; \node at (7,-2.8) {$x_{r-1,n'}$};

\node at (-0.1,-4.8) {$x_{r,1}$}; \node at (1.9,-4.8) {$x_{r,r-1}$};
\node at (4.9,-4.8) {$x_{r,n'-r+2}$}; \node at (6.9,-4.8) {$x_{r,n'}$};

\node at (0.9,-0.6) {$y_{1,r-1}$}; \node at (5.9,-0.6) {$y_{1,n'}$};

\node at (1,-3.6) {$y_{r-1,r-1}$}; \node at (6,-3.6) {$y_{r-1,n'}$};

\node at (3.4,-5.8) {$y_0$};

\node at (-1.8,0.2) {$X_1$}; \node at (-1.8,-2.8) {$X_{r-1}$};
\node at (-1.8,-4.9) {$X_r$};

\node at (2,-1.1) {$e_{1,r-1}$}; \node at (6.8,-1.1) {$e_{1,n'}$};
\node at (2.1,-4.1) {$e_{r-1,r-1}$}; \node at (6.9,-4.1) {$e_{r-1,n'}$};

\tikzstyle{every node}=[draw,circle,fill=white,minimum size=6pt,inner sep=0pt]

\foreach \numbertwo in {0,-3,-5}
{
\foreach \numberone in {0,2,5,7}
{
\node at (\numberone,\numbertwo) {};
}
}

\foreach \numbertwo in {-0.8,-3.8}
{
\foreach \numberone in {1,6}
{
\node at (\numberone,\numbertwo) {};
}
}

\node at (3.5,-6) {};

\foreach \numberone in {0,4.9}
{
\foreach \numbertwo in {0,-3}
{
\draw [thick,rounded corners=16pt] (\numberone,\numbertwo - 0.5) -- (\numberone-1, \numbertwo + 0.5) -- (\numberone + 1,\numbertwo + 0.5);
\draw [thick,rounded corners=16pt] (\numberone + 1,\numbertwo + 0.5) -- (\numberone + 3,\numbertwo + 0.5) -- (\numberone + 2,\numbertwo - 0.5);
\draw [thick,rounded corners=16pt] (\numberone,\numbertwo - 0.5) -- (\numberone + 1,\numbertwo - 1.5) -- (\numberone + 2,\numbertwo - 0.5);
}
}

\draw [help lines,rounded corners=17pt] (-1.5,0.6) rectangle (7.5,-0.3);
\draw [help lines,rounded corners=17pt] (-1.5,-2.4) rectangle (7.5,-3.3);
\draw [help lines,rounded corners=17pt] (-1.5,-4.5) rectangle (7.5,-5.4);

\end{tikzpicture}
\caption{Illustration of the edge set $\mathcal{E}'$ in Lemma \ref{lem:nphard}.
Edges of $\mathcal{E}'$ are in black. The sets $X_1, \dots X_r$ are in gray.
Edges of $\mathcal{E}(G)$ (not depicted) contain one vertex from each of  $X_1, \dots X_r$.}\label{fig:lem3}
\end{figure}


\begin{lemma}\label{lem:test2cover}
{\sc Test-$2$-Cover}$(2(n-1)/3+k)$ is NP-complete for $k=0$.
\end{lemma}

\begin{proof}
 We will show NP-completeness of our problem by reduction from the $P_3$-packing problem,
which asks,
given a graph $G$ with $n$ vertices, $n$ divisible by $3$, whether $G$ has $n/3$ vertex-disjoint copies of $P_3$. The problem is NP-complete \cite{GareyJohnson}.
Consider the graph $H$ obtained from $G$ by adding to it an isolated vertex $x$. We will view $H$ as an instance of {\sc Test-$2$-Cover} and
we will prove that $G$ is a {\sc Yes}-instance of the $P_3$-packing problem if and only if $H$ has a test cover of size $2n/3$.
If $G$ contains a set of edges $F$ forming $n/3$ disjoint copies of $P_3$, then observe that $F$ forms a test cover in $H$ of size $2n/3$.
Now assume that $H$ has a test cover $F$ of size $2n/3$. We will prove that $G$ is a {\sc Yes}-instance of the $P_3$-packing problem.
Observe that every vertex of $G$ has positive degree in $G[F]$. Let $n'$ be the number of vertices of degree 1 in $G[F]$.
We know that no vertex can have degree larger than 2, as otherwise  $F$ must have more than $2n/3$ edges.
 Since $|F|=2n/3$, we have $n'+2(n-n')= 2|F|=4n/3$ and so $n'=2n/3.$ Finally, $F$ cannot have any isolated edges as it is a test cover. Thus, $F$ is a collection of $n/3$ vertex-disjoint copies of $P_3$.
\end{proof}

The next theorem follows from the straightforward fact that {\sc Test-$r$-Cover}$(\frac{2(n-1)}{r+1}+k)$ is in para-NP
and the previous two lemmas.

\begin{theorem}
{\sc Test-$r$-Cover}$(\frac{2(n-1)}{r+1}+k)$ is para-NP-complete for each fixed $r\ge 2$.
\end{theorem}

\section{A short FPT proof for {\sc Test-$r$-Cover}$(m-k)$}\label{sec:m-kfpt}
In this section, we show that {\sc Test-$r$-Cover}$(m-k)$  is fixed-parameter tractable. The proof is short but does not lead to a polynomial kernel; this requires more work, and is the subject of the next section.

\begin{theorem}\label{thm:m-kfpt}
 There is an algorithm for {\sc Test-$r$-Cover}$(m-k)$ that runs in time $(r^2+1)^k(n+m)^{O(1)}$.
\end{theorem}

\begin{proof}
We first need to guess whether there will be a vertex not contained in any edge in the solution, and if so, which vertex it will be. If there already exists a vertex $y_0$ not in any edge in $\mathcal{E}$, then we are done. Otherwise, either pick a vertex $x$ or guess that every vertex in $V$ will be covered by the solution. If a vertex $x$ is picked, delete all the edges containing $x$, and reduce $k$ by the number of deleted edges. If it is guessed that every vertex in $V$ will be covered by the solution, add a new vertex $y_0$ which is not in any edge. Observe that this does not change the solution to the problem.
By doing this we have split the problem into $n+1$ separate instances, with each instance containing an isolated vertex.
Thus we may now assume that there exists a vertex $y_0$ which is not contained in any edge in $\mathcal{E}$.

Consider an edge $e \in \mathcal{E}$, and suppose that $\mathcal{E}\setminus \{e\}$ is a test cover.
Let $\mathcal{B}_0$ be a minimal set of edges in $\mathcal{E}\setminus \{e\}$ which covers $e$. Note that such a set must exist, as otherwise $x$ is not separated from $y_0$ in $\mathcal{E}\setminus \{e\}$  for some $x \in e$, and so $\mathcal{E}\setminus \{e\}$ is not a test cover. Furthermore, we may assume $|\mathcal{B}_0|\le |e| \le r$.
Now for each $b \in \mathcal{B}_0$, let $\mathcal{B}_b$ be a minimal set of edges in $\mathcal{E}\setminus \{e\}$ separating every vertex in $b \setminus e$ from every vertex in $b \cap e$. 
By Corollary \ref{cor:sep},
we may assume that $|\mathcal{B}_b|\le r-1$.

Now let $\mathcal{B}=\mathcal{B}_0 \cup (\bigcup_{b \in \mathcal{B}_0} \mathcal{B}_b)$, and observe that $\mathcal{B}$  isolates the vertex set of $e$. Thus, in any solution with minimum number of edges, at least one edge from $\mathcal{B} \cup \{e\}$ will be missing. Note that $|\mathcal{B}| \le r + r(r-1) = r^2$.

We now describe a depth-bounded search tree algorithm for {\sc Test-$r$-Cover}$(m-k)$.
If $\mathcal{E}$ is not a test cover, return {\sc No}. Otherwise if $k=0$ return {\sc Yes}.
Otherwise, for each edge $e \in \mathcal{E}$ check whether $\mathcal{E} \setminus \{e\}$ is a test cover.
If for all $e \in \mathcal{E}$, $\mathcal{E} \setminus \{e\}$ is not a test cover,
 then a test cover must contain all $m$ edges and so we return {\sc No}.
Otherwise, let $e$ be an edge such that $\mathcal{E} \setminus \{e\}$ is a test cover, and construct the set $\mathcal{B}$ as described above. Then we may assume one of $\mathcal{B} \cup \{e\}$  is not in the solution. Thus we may pick one edge from $\mathcal{B} \cup \{e\}$, delete it, and reduce $k$ by $1$. So we split into $r^2+1$ instances with reduced parameter.

We therefore have a search tree with at most $(r^2+1)^k$ leaves.  As every internal node has at least $2$ children, the total number of nodes is at most $2(r^2+1)^k-1$. Note also that guessing the isolated vertex at the start split the problem into $n+1$ instances, so there are at most $(n+1)(2(r^2+1)^k-1)$ nodes to compute in total. As each node in the tree takes polynomial time to compute, we have an algorithm with total running time  $(r^2+1)^k(n+m)^{O(1)}$.
\end{proof}

\section{Polynomial Kernel for {\sc Test-$r$-Cover}$(m-k)$}\label{sec:m-k}
In this section, we show that {\sc Test-$r$-Cover}$(m-k)$ admits a polynomial kernel.
It will be useful to consider a slight generalization of {\sc Test-$r$-Cover}$(m-k)$, which we call {\sc Subset-Test-$r$-Cover}$(m-k)$. In this problem, we are given a special subset $\mathcal{B} \subseteq \mathcal{E}$ of edges which are required to be in the solution. For convenience we will say these edges are \emph{colored black}. 

%

We begin with the following reduction rules, which must be applied whenever possible:


\begin{krule}\label{rule:blackedge1}
Given a vertex $x$ of degree 1 and a black edge $b$
which contains only $x$, delete $b$ from $\cal B$ and $\cal E$, delete $x$
and leave $k$ the same.
\end{krule}

\begin{krule} \label{rule:blackedge2}
Given a black edge $b$, if there exists any other edge $e$ such that $b\subset e$,
then replace $e$ with $e\setminus b$. If there exists a black edge $b'$ such that
$b$ cuts $b'$ and $b'$ cuts $b$,
delete $b$ and $b'$ and add black edges $b\setminus b'$, $b'\setminus b$ and $b\cap b'$. Leave $k$ the same.
\end{krule}


\begin{lemma}\label{lem:blackedgerules}
 Let $(V,\mathcal{E}', \mathcal{B}',k)$ be an instance of  {\sc Subset-Test-$r$-Cover}$(m-k)$ derived from  $(V,\mathcal{E}, \mathcal{B}, k)$ by an application of Rule \ref{rule:blackedge1} or \ref{rule:blackedge2}. Then $(V,\mathcal{E}', \mathcal{B}', k)$ is a {\sc Yes}-instance if and only if $(V,\mathcal{E}, \mathcal{B}, k)$ is a {\sc Yes}-instance.
\end{lemma}

\begin{proof}
We will show for each rule that for any $t$, $(V,\mathcal{E}, \mathcal{B}, k)$ has a solution of size $|\mathcal{E}|-k$ if and only if     $(V,\mathcal{E}', \mathcal{B}', k)$ has a solution of size $|\mathcal{E}'|-k$.

 {\bf Rule \ref{rule:blackedge1}: } Suppose $(V,\mathcal{E}', \mathcal{B}', k)$ is a {\sc Yes}-instance, with solution ${\cal T}'$. Then observe that ${\cal T} = {\cal T}' \cup \{b\}$ is a solution for $(V,\mathcal{E}, \mathcal{B}, k)$.
Conversely, if ${\cal T}$ is a solution for $(V,\mathcal{E}, \mathcal{B}, k)$ then ${\cal T}$ must contain $b$, and ${\cal T}\backslash \{b\}$ is a solution for $(V,\mathcal{E}', \mathcal{B}', k)$.

{\bf Rule \ref{rule:blackedge2}: } First consider the case when $b \subset e$ for some other edge $e$. It is sufficient to show that for any ${\cal T} \subseteq \mathcal{E}$ containing $e$ and $b$, ${\cal T}$ is a test cover if and only if $({\cal T} \setminus \{e\}) \cup \{e \setminus b\}$ is a test cover. To see this, observe that for any $x \in e, y \notin e$, $x$ and $y$ are separated either by $b$ or $e \setminus b$, and for any $x \in e \setminus b$, $y \notin  e \setminus b$, $x$ and $y$ are separated either by $b$ or $e$.

Now consider the case when $b, b'$ are intersecting black edges.
Similar to the previous case, if $x$ and $y$ are separated by one of $b,b'$ then they are also separated by at least one of $b\setminus b'$, $b'\setminus b$, $b\cap b'$, and if they are separated by one of $b\setminus b'$, $b'\setminus b$, $b\cap b'$ then they are also separated by at least one of $b,b'$.
\end{proof}

\begin{lemma}\label{lem:degree}
Let $(V,\mathcal{E}, \mathcal{B},k)$ be an instance irreducible by  Rules \ref{rule:blackedge1} and \ref{rule:blackedge2}.  The instance can be reduced, in polynomial time, to an equivalent instance such that every vertex has degree at most $kr^2$.
\end{lemma}
\begin{proof}
Assume that there exists a vertex $x$ with degree in $\mathcal{E}$ greater than $kr^2$.
We will be 
able to produce an equivalent instance in which either $k$ or the degree of $x$ is 
reduced.
Clearly this reduction can only take place a polynomial number of times, so in polynomial time we will reduce to an instance in which every vertex has degree bounded by $kr^2$.

We produce a special set $\widetilde{X}$, such that $\widetilde{X}$ is still isolated when at most $k$ edges are deleted, according to the following algorithm.

\RestyleAlgo{algoruled}
\begin{algorithm}[H]
Set $\widetilde{\mathcal{E}} = \mathcal{E}$, $i=1$, $X=\{x\}$, $j=1$\;
 \While{$i \le k+1$}{
  \eIf{$\widetilde{\mathcal{E}}$ isolates $X$}
  {
  Let $e_i$ be an edge containing $X$, and construct a set $E_i\subseteq\widetilde{\mathcal{E}}$ such that $E_i\cup \{e_i\}$
isolates $X$ and $|E_i| \le r-1$\;
Set $\widetilde{\mathcal{E}} = \widetilde{\mathcal{E}} \setminus (E_i\cup \{e_i\})$ \;
  Set $i=i+1$\;    
    }{
   Let $X'$ be the class induced by $\widetilde{\mathcal{E}}$ containing $X$\;
   Set $X = X', i = 1, j=j+1$;
   }
 }
Set $\widetilde{X} = X$.
\end{algorithm}

Observe that throughout the algorithm, by construction any edge in  $\widetilde{\mathcal{E}}$ which contains $x$ also contains $X$ as a subset, $|X|\ge j$ and $\widetilde{\mathcal{E}}\subseteq \mathcal{E}$. Notice that if the algorithm ever sets $j=r+1$, then at that point at most $kr^2$ edges have been deleted from  $\widetilde{\mathcal{E}}$, and as $|X|\ge r+1$, no remaining edges in $\widetilde{\mathcal{E}}$ contain $x$. But this is a contradiction as the degree of $x$ is greater than $kr^2$. Therefore we may assume the algorithm never reaches $j=r+1$. Hence the algorithm must terminate for some $j\le r$.

We now show that we can always find $e_i$ and $E_i$ for $i\le k+1$.
Since $x$ has degree greater than $kr^2$ and at most $kr(r-1)$ edges are removed from $\widetilde{\mathcal{E}}$ earlier in the algorithm, we can always find an edge $e_i$ containing $x$ and therefore containing $X$. To see that $E_i$ can be constructed using at most $r-1$ edges, apply Corollary \ref{cor:sep} to $X$ and $e_i \setminus X$. 

Now consider the set $\widetilde{X}$ formed by the algorithm. When the algorithm terminated, $e_i$ and $E_i$ were found for all $i\le k+1$. If we remove $k$ arbitrary edges from $\mathcal{E}$,
it is still possible to find $i$ such that no edges in $E_i\cup \{e_i\}$ have been deleted. This means that as long as we delete at most $k$
edges, $\widetilde{X}$ is still isolated.
Therefore if $\widetilde{X}$ is an edge in $\mathcal{E}$, delete it and reduce $k$ by $1$.
If $\widetilde{X}$ is not an edge in $\mathcal{E}$, add a new black edge $\widetilde{X}$ to $\mathcal{E}$ and $\mathcal{B}$, keeping $k$ the same, and apply Rule \ref{rule:blackedge1} and \ref{rule:blackedge2}.
Observe that since $\widetilde{X}$ is properly contained in at least two edges, before adding the black edges, this will decrease the degree of every vertex in $\widetilde{X}$.
\end{proof}

Now assume that $(V,\mathcal{E}, \mathcal{B},k)$ is reduced by Rules \ref{rule:blackedge1} and \ref{rule:blackedge2} and that every vertex has degree at most $kr^2$.
We will color the uncolored edges in $\mathcal{E}$ as follows.
%
For every edge $e$ which is not black, if $\mathcal{E}\setminus \{e\}$ is not a test cover, color $e$ black, adding it to $\mathcal{B}$
(and apply Rules \ref{rule:blackedge1} and \ref{rule:blackedge2}). If $\mathcal{E}\setminus \{e\}$ is a test cover and $e$ contains a degree one vertex, color $e$ orange.
Otherwise, color $e$ green.

\begin{rem}\label{isol}
Notice that an edge is colored orange only if 
 there is no isolated vertex.
\end{rem}

\begin{lemma}\label{green}
If $G$ is a set of green edges such that, for every pair $g_1,g_2\in G$, $N_1[g_1]\cap N_1[g_2]$ is empty, then $\mathcal{E}\setminus G$
is a test cover.
\end{lemma}
\begin{proof}
We proceed by induction on $|G|$. If $|G|=1$ this is obviously true. If $|G|=j+1$, delete the first $j$ edges and consider the last one, denoted $g$. The only problem
that could occur removing $g$ is that a vertex $x\in g$ may no longer be separated from another vertex $y$. If $y$ is not in one of
the edges in $G \setminus \{g\}$,
then $x$ and $y$ are not separated even by $\mathcal{E}\setminus \{g\}$,
which is a contradiction since $g$ is green.
Therefore, denote by $g'$ the edge in $G$ which contains $y$. The degree of $y$ is at least $2$, hence there exists an edge different from $g'$ which contains
$y$; this edge cannot contain $x$ too, or $N_1[g]\cap N_1[g']$ would not be empty. This ensures that 
$x$ and $y$ are separated by $\mathcal{E}\setminus G$.
\end{proof}

\begin{krule}\label{rule:orangeedge}
 Given an orange edge $o$, if $N_2[o]$ contains no green edges, delete $o$ and decrease $k$ by $1$.
(Notice that this creates an isolated vertex, which means that every other orange edge will become black.)
\end{krule}


\begin{lemma}\label{orange}
 Let $(V,\mathcal{E}', \mathcal{B}',k-1)$ be an instance of  {\sc Subset-Test-$r$-Cover}$(m-k)$ derived from  $(V,\mathcal{E}, \mathcal{B}, k)$ by an application of Rule \ref{rule:orangeedge}. Then $(V,\mathcal{E}', \mathcal{B}', k-1)$ is a {\sc Yes}-instance if and only if $(V,\mathcal{E}, \mathcal{B}, k)$ is a {\sc Yes}-instance.
\end{lemma}

\begin{proof}
Let $(V,\mathcal{E}, \mathcal{B},k)$ be a {\sc Yes}-instance, and 
suppose there is an orange edge $o$ such that $N_2[o]$ contains no green edges.
It is sufficient to prove that
there exists a solution that does not include $o$.

Suppose $\mathcal{T}\subseteq\mathcal{E}$ is a solution and suppose it contains $o$. If there is a vertex $x$ which is not contained
in any edge of $\mathcal{T}$, consider an edge $e\in\mathcal{E}$ which contains it (which exists by Remark \ref{isol}). If the isolated
vertex $x$ does not exist, take any edge $e$ from $\mathcal{E}\setminus\mathcal{T}$.

Consider $\mathcal{T}'=(\mathcal{T}\setminus\{o\})\cup\{e\}$. We claim that $\mathcal{T}'$ is still a test cover.

First of all, note that if there is an orange edge $o'\in \mathcal{E}\setminus\mathcal{T}$, this edge must be $e$:
removing $o'$ creates an isolated vertex and $o'$ is the only edge containing that vertex, which means that $o'$ is the edge
that we add 
 to make $\mathcal{T}'$.
Therefore, we may assume $(N_2[o]\setminus\{o\})\subseteq\mathcal{T}'$.

Now, the only problem that can occur removing $o$ is that a vertex $x\in o$
is no longer separated from a vertex $y\in V\setminus o$.  Vertices $x$ and $y$ must be separated by some other edge in $\mathcal{E}\setminus\mathcal{T}'$ as $o$ is not black; furthermore this edge must contain $y$ and not $x$, as any other edge containing $x$ is in $(N_2[o]\setminus\{o\})\subseteq\mathcal{T}'$.
 Let this edge
be $\widetilde{e}$.

If $x$ is the degree $1$ vertex, then it is now the only isolated vertex and therefore it is separated
from any other vertex. If $x$ is a vertex of degree at least $2$, then there is an edge
(different from $o$) containing it; moreover, this edge cannot contain any vertices of $\widetilde{e}$, because in this case
$\widetilde{e}\in N_2[o]$. Hence, even deleting $o$, every vertex $x \in o$ is still separated from any other vertex,
which ensures that $\mathcal{T}'$ is a test cover.
\end{proof}


\begin{lemma}\label{neigh}
 Let $(V,\mathcal{E}, \mathcal{B},k)$ be an instance obtained after applying Rules \ref{rule:blackedge1}, \ref{rule:blackedge2} and \ref{rule:orangeedge}, and let $x$ be a non-isolated vertex. Then
 $x\in V(N_3[g])$ for some green edge $g$.
\end{lemma}
\begin{proof}
If $x$ is contained in a green edge, we are done. If $x$ is contained in an orange edge $o$, Lemma \ref{orange} implies that there
exists a green edge $g$ in $N_2[o]$, which means that $x\in V(N_2[g])\subseteq V(N_3[g])$. If, finally, $x$ is contained in a black edge $b$,
this edge must intersect one other edge, that can be either green or orange (due to Rules \ref{rule:blackedge1} and \ref{rule:blackedge2}).
In both cases, $x\in V(N_3[g])$ for some green edge $g$.
\end{proof}

\begin{theorem}
There is a kernel for {\sc Subset-Test-$r$-Cover}$(m-k)$  with $|V|\leq (k-1)k^5r^{16}+1$ and $|\mathcal{E}|\leq (k-1)k^5r^{16}+k$. This gives a kernel for {\sc Test-$r$-Cover}$(m-k)$ with $|V|\le 5(k-1)k^5r^{16}+4k+1$ and $|\mathcal{E}|\le 3(k-1)k^5r^{16}+3k$.
\end{theorem}
\begin{proof}

 Let $(V,\mathcal{E}, \mathcal{B},k)$ be an instance irreducible by Rules \ref{rule:blackedge1}, \ref{rule:blackedge2} and \ref{rule:orangeedge}.
Construct greedily a set $G$ of green edges which satisfy the hypothesis of Lemma \ref{green}. If $|G|\geq k$, we answer {\sc Yes} using Lemma \ref{green}.
Otherwise, $|G|\leq k-1$ and every green edge which is not in $G$ must be in $N_2[G]$. By Lemma \ref{neigh},
this means that every vertex (except the one of degree zero, if it exists)
must be in $V(N_3[N_2[G]])=V(N_5[G])$.

However, $|V(N_5[G])|\leq r|N_5[G]|$ and, given $\mathcal{F}\subseteq\mathcal{E}$, $|N_1[\mathcal{F}]|\leq|\mathcal{F}|(r)(kr^2)$ (by Lemma \ref{lem:degree}),
which means that $|N_5[G]|\leq|G|(kr^3)^5$.
To sum up, $|V(N_5[G])|\leq (k-1)k^5r^{16}$, which gives us the required bound on the number of vertices.

To bound the number of edges, we show that there is a solution of size at most $|V|$. First let $\mathcal{T}$ be the set of black edges. By Rule \ref{rule:blackedge2}, the black edges are disjoint and therefore $|\mathcal{T}|\le |V|$ and $\mathcal{T}$ induces at least $|\mathcal{T}|$ classes. Now if $\mathcal{T}$ is not a test cover, add an edge to $\mathcal{T}$ that increases the number of induced classes. Then eventually we have that $|\mathcal{T}|\le |V|$ and $\mathcal{T}$ induces $|V|$ classes as required.
Therefore if $|\mathcal{E}|-k\ge |V|$, the answer is {\sc Yes}. Hence $|\mathcal{E}| \le |V|+k-1 \le (k-1)k^5r^{16}+k$,
proving the kernel for {\sc Subset-Test-$r$-Cover}$(m-k)$.

We now prove the kernel for {\sc Test-$r$-Cover}$(m-k)$. First transform an instance  $(V, \mathcal{E}, k)$ into an equivalent instance  $(V, \mathcal{E}, \mathcal{B}, k)$ of {\sc Subset-Test-$r$-Cover}$(m-k)$, by letting $\mathcal{B}=\emptyset$. Then reduce this instance to a kernel $(V', \mathcal{E}', \mathcal{B}', k')$.
Now we reduce $(V', \mathcal{E}', \mathcal{B}', k')$ to an instance $(V'', \mathcal{E}'', k'')$ of {\sc Test-$r$-Cover}$(m-k)$, completing the proof.
If an edge $b$ is colored black, and ${\cal E}\setminus \{b\}$ is not a test cover, then uncolor $b$. Otherwise, observe that by construction, $b$ must have been created during the algorithm of Lemma \ref{lem:degree}, and in such a case $b$ was contained within $k''+1$ other edges. Therefore $b$ contains at most $r-1$ vertices. We will replace $b$ with a small gadget that is `equivalent' to $b$.


To make the gadget, add vertices
 $x_1,x_2,x_3,x_4$ to the instance, replace $b$ with $b\cup \{x_1\}$ and add edges $e'=\{x_1,x_2,x_3\}$ and
$e''=\{x_3,x_4\}$. Observe edge $e'$ is necessary to separate $x_3$ from $x_4$, $e''$ is necessary to separate $x_2$ from $x_3$ and $b\cup
\{x_1\}$ is necessary to separate $x_1$ from $x_2$.
Hence all three edges must be in a test cover. Hence, in any test cover solution for the new instance, if we replace the gadget with $b$,
 we obtain a test cover solution for the original instance. Furthermore, such a solution will contain $b$. In this sense, the problems are equivalent.

Finally observe that for each original black edge we added four new vertices and two new edges. Hence $|V''|\le |V'|+4|\mathcal{E}'| \le 5(k-1)k^5r^{16}+4k+1$ and $|\mathcal{E}''|\le |\mathcal{E}'|+2|\mathcal{E}'| \le 3(k-1)k^5r^{16}+3k$.
\end{proof}

%

\section{Polynomial kernel for {\sc Test-$r$-Cover}$(n-k)$}\label{sec:n-k}
In this section, we show that {\sc Test-$r$-Cover}$(n-k)$ admits a kernel which is polynomial in the number of vertices. 
We may assume that $k\ge 2$ as  the answer  to {\sc Test-$r$-Cover}$(n-k)$ for $k\le 1$ is always positive.


We reuse some terminology from \cite{CrowstonGJSY12}.
We say $\mathcal{T} \subseteq \mathcal{E}$ is a \emph{$k$-mini test cover} if $|\mathcal{T}|\le 2k$ and the number of classes induced by $\mathcal{T}$ is at least $|\mathcal{T}|+k$.

The following result is a consequence of Lemma 1 and Theorem 2 of \cite{CrowstonGJSY12}.
\begin{lemma}\label{mini}
Assume $\mathcal{E}$ is a test cover of $V$. Then the following are equivalent:
\begin{enumerate}
 \item $\mathcal{E}$ contains a test cover of size at most $n-k$.
 \item There exists $\mathcal{T}\subseteq \mathcal{E}$ such that the number of classes induced by $\mathcal{T}$ is at least $|\mathcal{T}|+k$.
 \item $\mathcal{E}$ contains a $k$-mini test cover.
\end{enumerate}
\end{lemma}

The following result follows from Lemma 1 and the proof of Theorem 2 in \cite{CrowstonGJSY12}.

\begin{lemma}\label{lem:findF}
 In polynomial time, we may either find a $k$-mini test cover, or find an $\mathcal{F} \subseteq \mathcal{E}$ such that:
\begin{enumerate}
 \item $|\mathcal{F}| < 2k$
 \item $\mathcal{F}$ induces less than $|\mathcal{F}|+k$ classes.
  \item Each edge in $\mathcal{E}$ cuts at most one class induced by $\mathcal{F}$.
\item  For any $e,e' \in \mathcal{E}$ and any class $K$ induced by
$\mathcal{F}$,at least one of $(e\cap e')\cap K$,
$(e \backslash e')\cap K$, $(e \backslash e')\cap K$ and $K \backslash(e \cup e')$ is empty.
\end{enumerate}
\end{lemma}

%

Now assume we have found such an $\mathcal{F}$ and let $C_1,\dots,C_l,G$ be the classes induced by $\mathcal{F}$ 
(where $G$ is the class of vertices which are not contained in any edge of $\mathcal{F}$ and $l<3k$).
Let $\mathcal{C}$ be the set of classes $C_1,\dots,C_l$, and let $C$ be the set of vertices contained in such classes.
Let $\mathcal{G}$ be the set of edges that intersect $G$.
For each edge  $e \in \mathcal{G}$, we say $e \cap G$ is the \emph{$G$-portion} of $e$.
A subset $\Gamma$ of $G$ is a \emph{component} if $\Gamma$ is the $G$-portion of an edge $e \in  \mathcal{G}$ and $\Gamma \not\subset e'\cap G$ for all edges $e'\in \mathcal{G}$.
Notice that the number of vertices in $C$ is bounded by $(2k-1)r$ because these vertices are contained in edges
of $\mathcal{F}$ and every edge contains at most $r$ vertices. Also notice that every component of $G$ has at most $r$ vertices.

%


\begin{theorem}\label{thm:thm8}
Given an instance $(V,\mathcal{E},k)$, it is possible to reduce it in polynomial time to an equivalent instance with at most
$18k^3r$
 vertices and 
$(18k^3r)^r$
 edges.
\end{theorem}
\begin{proof}
We may assume that $|V|>(7k+2)r$ as otherwise our instance is already a kernel (every edge has at most $r$ vertices and so 
the number of edges is at most $((7k+2)r)^r$).
Observe that if there exists an edge $e$ such that $|e \cap G| \ge |G|/2$ then $|G|\le 2|e|\le 2r$, and we conclude that
$|V| \le |C|+2r \le (2k+1)r$, 
a contradiction. Therefore, $|e \cap G| < |G|/2$, and so by Part 4 of Lemma \ref{lem:findF} if 
$X, Y$ are different $G$-portions then either 
 $X \subset Y$, $Y \subset X$ or $X \cap Y = \emptyset$.

Apply the following algorithm:

\begin{myindentpar}{0.5cm}
\textbf{Step 1:} For each pair $(C_i,C_j)$ ($i\neq j$) in turn, mark $2k$ unmarked components of 
$G$ which contain the $G$-portion of an edge containing $C_i$ and having empty intersection with $C_j$ (mark these edges too). If there are less than 
 $2k$ such components, mark them all. Let $E_{i,j}$ denote the set of marked edges.

For each $C_i$ in turn, mark $2k+1$ unmarked components of 
$G$ which contain the $G$-portion of an edge containing $C_i$ (mark these edges too). If there are less than
$2k+1$ such components, mark them all. Let $E_i$ denote the set of marked edges.

\textbf{Step 2:} Delete every edge in $\mathcal{G}$ whose $G$-portion is not contained in a marked component of $G$. 
Delete every vertex which is not contained in any edge anymore, except one vertex $y$ (if it exists).
\end{myindentpar}

Let $\mathcal{E}'$ be the set of edges which have not been deleted by this algorithm. Notice that the number of marked components in $G$ is at most
$(3k-1)(3k-2)(2k) + (3k-1)(2k+1) < (3k-1)^2(2k+1) = 18k^3-3k^2-4k+1 < 18k^3-2k$
(here we use the assumption that $k \ge 1$).
 Let $G'$ be the set of vertices of $G$ which have not been deleted
by the algorithm. Notice that 
$|G'|\leq (18k^3-2k)r$.

In the instance which is produced, 
$|V'|=|C|+|G'| \leq 18k^3r$, and since each edge contains at most $r$ vertices,
$|\mathcal{E'}| \le (18k^3r)^r$. 
Hence it is sufficient to show that $(V',\mathcal{E}',k)$ admits a $k$-mini test cover if and only if $(V,\mathcal{E},k)$ admits one.

Obviously, if $(V',\mathcal{E}',k)$ admits a $k$-mini test cover, this is a $k$-mini test cover for $(V,\mathcal{E},k)$ too.
For the other direction, suppose $\mathcal{T}$ is a $k$-mini test cover for $(V,\mathcal{E},k)$ such that $\mathcal{T}\setminus \mathcal{E}'$ is as small as possible.
For the sake of contradiction, suppose that $\mathcal{T}$ contains at least one edge $e$ in $\mathcal{T} \setminus \mathcal{E}'$.
We claim that it is possible to construct a set $\mathcal{T}'''$ which induces at least $|\mathcal{T}'''|+k$ classes, 
such that
${\cal T}''' \setminus{\cal E}' = ({\cal T} \setminus{\cal E}')\setminus \{e\}$.
By applying Lemma \ref{mini} to the hypergraph $H'$ with edge set $\mathcal{T}'''$, and vertex set formed by identifying the vertices in each class induced by $\mathcal{T}'''$, 
observe that there exists a $k$-mini test cover in $H'$. Observe that the edges of this $k$-mini test cover in $H'$ also form a $k$-mini test cover in the original instance, and this $k$-mini test cover is a subcollection of  $\mathcal{T}'''$. Since this $k$-mini test cover contains fewer edges from 
$\mathcal{E} \setminus \mathcal{E'}$ than $\mathcal{T}$ does, we have a contradiction.

Start with $\mathcal{T}'=\mathcal{T}\setminus\{e\}$.
Since $e$ is not in $\mathcal{E}'$, $e$ must be in $\mathcal{G}$, and the $G$-portion of $e$ must not be contained in any marked component. 
Furthermore, for each $C_i,C_j \in {\cal C}$ with $C_i \subseteq e$ and $e \cap C_j = \emptyset$ we note that $E_{i,j}$ must contain  $2k$ edges, as
otherwise $e$ would be in $\mathcal{E}'$. 
 Similarly, for each $C_i$ contained in $e$ we note that $E_i$ must contain $2k+1$ edges.
For any $i,j$ such that $|E_{i,j}|=2k$, let $e_{i,j}$ be an edge in $E_{i,j}$ 
whose $G$-portion is disjoint  
from any edge in $\mathcal{T}'$. 
This must exist as
 $|{\cal T}'| \le 2k-1$. 

For any $i$ such that $|E_i|=2k+1$ let $e_i, e_i'$ be edges 
in $E_i$ whose $G$-postions 
are disjoint from any edge in $\mathcal{T}'$. 
These edges must exist as $|{\cal T}'| \le 2k-1$.

Let $C^*_0$ be the class induced by $\mathcal{T}'$ that consists of all vertices not in any edge in ${\mathcal T}'$ (which exists by Claim C below).
We will need the following claims.

\2

\begin{description}
\item[Claim A:] There is at most one class $C_G^*$ induced by $\mathcal{T}'$, such that $G \cap (C_G^* \cap e) \not= \emptyset$ and
$G \cap (C_G^* \setminus e) \not= \emptyset$.

{\em Proof of Claim A:}  For the sake of contradiction assume that there are two such classes $C_G'$ and $C_G''$. 
This implies that there exist vertices $x' \in G \cap (C_G' \cap e)$,
$y' \in G \cap (C_G' \setminus e)$, $x'' \in G \cap (C_G'' \cap e)$ and $y'' \in G \cap (C_G'' \setminus e)$. 
Some edge $e' \in \mathcal{T}'$ separates $C_G'$ and $C_G''$.
Note that adding $e'$ and $e$ to ${\cal F}$ separates $x'$, $y'$, $x''$ and $y''$ into different classes, contradicting Part 4 of Lemma~\ref{lem:findF}.  
This contradiction complete the proof of Claim A.

\2

\item[Claim B:] For each edge $e'$ that cuts $G$ and every $C_i$ we have $C_i \subseteq e'$ or $C_i \cap e'=\emptyset$. In particular, $C_i \subseteq e$ or $C_i \cap e=\emptyset$.

{\em Proof of Claim B:}  If Claim~B is false then there exist $x \in C_i \cap e'$ and $y \in C_i \setminus e'$.
So adding $e'$ to ${\cal F}$ cuts $G$ (as $G$ is not a subset of any edge and $e'$ contains vertices from $G$)  and $C_i$, a contradiction to Part 3 in Lemma~\ref{lem:findF}.

\2

\item[Claim C:] $C^*_0$ exists and $|G \cap C_0^*| \geq (3k+2)r$.

{\em Proof of Claim C:} If $C^*_0$ does not exist then every vertex of $V$ belongs to some edge in ${\cal T}'$, which implies that $|V| \leq 2kr$, so 
$C^*_0$ does exist. If  $|G \cap C_0^*| < (3k+2)r$, then the following holds and we have a contradiction  to the assumption on $|V|$ in the beginning of the proof:

$|V| \leq 2kr + |C_0^*| \leq 2kr + (|C| + |G \cap C_0^*|) <  2kr + 2kr + (3k+2)r = (7k+2)r.$
\end{description}

By Claim A there exists at most one class, say $C_G^*$, induced by $\mathcal{T}'$ that is cut by $e$ and only contains vertices from $G$.
Let $C^*_1, \dots, C^*_t$ be all classes induced by $\mathcal{T}'$, different from $C_G^*$ and $C_0^*$, that are cut by $e$. Note that $t \leq 3k$, as ${\cal T}$ is a $k$-mini test cover.
Each $C^*_s$, ($1 \leq s \leq t$), must be contained in an edge, say $e_s^*$, in $\mathcal{T}'$ and contain vertices from  $C$, by the definitions on $C_G^*$ and $C_0^*$.
We are going to create a collection of edges $\mathcal{T}''$ 
such that each $C^*_s$ is cut by an edge in $\mathcal{T}''$  
and also $\mathcal{T}''$ induces $|\mathcal{T}''|$ extra classes
 in $C_0^*$. Initially let $\mathcal{T}''=\emptyset$.
For each $s\in [t]$ in turn, consider the following two cases.

\2

{\em Case 1:} For some $i\neq j$, $e$ contains $C_i$ but not $C_j$ and 
$C_i \cap C_s^* \neq \emptyset, C_j \cap C_s^* \neq \emptyset$.

In this case 
observe that $|E_{i,j}| = 2k$, as otherwise $e$ would be marked. 
Then add the edge $e_{i,j}$ to $\mathcal{T}''$, if $e_{i,j}$ is not in $\mathcal{T}''$ already.
Note that $e_{i,j}$ separates $C_i$ from $C_j$ and 
therefore cuts $C_s^*$, and
also creates an 
extra class in $C_0^*$, as desired.

\2

{\em Case 2:} Case 1 does not hold. That is, $C \cap C_s^* \subseteq e$ or $(C \cap C_s^*) \cap e = \emptyset$.

  Recall that there exists a 
$C_i$ such that $C_s^*$ contains vertices from $C_i$
and, since $e$ cuts $C_s^*$, we have $C_s^*\cap G\neq\emptyset$.
Suppose $e$ does 
not contain $C_i$. Then $(C \cap C_s^*) \cap e = \emptyset$ and, since $e$ cuts $C_s^*$, it must contain vertices 
from $C_s^* \cap G \subseteq e_s^* \cap G$. Then  $e_s^*$ cuts $G$ and the $G$-portion of $e_s^*$ is in the same component as the $G$-portion of $e$, and therefore $e_s^*$ is an unmarked edge. Furthermore since $e_s^*$ cuts $G$ it does not cut $C_i$, and therefore $C_i \subseteq e_s^*$. Thus, we have that either $e$ or $e_s^*$ is an unmarked edge containing $C_i$, and therefore $|E_i|=2k+1$. Then add $e_i$ to $\mathcal{T}''$, if $e_i$ is not already in $\mathcal{T}''$.
Observe that $e_i$ cuts $C_s^*$ as it contains
vertices in $C_i \cap C_s^*$ but no
vertex from $C_s^* \cap G$, and $e_i$ creates an extra class in $C_0^*$, as required.

\2

This completes Case 1 and Case 2.
Note that the edges in $\mathcal{T}''$ all have vertex disjoint $G$-portions, as they are in
distinct $E_{i,j}$'s and $E_i$'s.
 We now consider Case~(i) and Case~(ii) below, which will complete the proof.

\2

{\em Case (i):} $C_G^*$ does not exist or is equal to $C_0^*$ or $e$ does not cut $C_0^*$ or does not cut  $C_G^*$.
 
In this case $e$ cuts at most $t+1$ classes induced by $\mathcal{T}'$. 
Note that if we add the 
edges from $\mathcal{T}''$ to $\mathcal{T}'$, each edge in $\mathcal{T}'$ increase the number of classes in 
$C_0^*$ by  
at least one. Also note that for every $s \in [t]$ some edge in $\mathcal{T}''$ cuts $C^*_s$.
So let $\mathcal{T}''' = \mathcal{T}' \cup \mathcal{T}'' = (\mathcal{T}\setminus{e}) \cup \mathcal{T}''$.
Removing $e$ from ${\cal T}$ decreases the number of classes by at most $t+1$ and
adding  $\mathcal{T}''$ increases the number of classes by at least $t+|\mathcal{T}''|$. So by increasing the number of edges by  $|\mathcal{T}''|-1$ we have increased the number of classes
by at least $|\mathcal{T}''|-1$ and therefore we still have at least $k$ more classes than edges.

\2

{\em Case (ii):} Case (i) does not hold. That is, $C_G^*$ exists and is distinct from $C_0^*$ and $e$ cuts both $C_G^*$ and $C_0^*$.

By Claim~A we note that $e$ either contains all of $C_0^* \cap G$ or none of $C_0^* \cap G$. By Claim~C  $e$ must contain none of 
$C_0^* \cap G$. As $e$ cuts $C_0^*$ 
we must have $C \cap e \cap C_0^* \not= \emptyset$.
Therefore there exists $C_i$ such that $e$ contains vertices from $C_i \cap C_0^*$,
and so $|E_i|=2k+1$. 
Add $e_i$ and $e_i'$ to $\mathcal{T}''$ (unless $e_i$ is already in $\mathcal{T}'$, 
 in which case just add $e_i'$ to $\mathcal{T}''$).
Observe that the $G$-portions of $e_i$ and $e_i'$ are vertex disjoint by
construction, and so the $G$-portions of all 
 edges in $\mathcal{T}''$ are still vertex disjoint.

Note that adding $e_i$ and $e_i'$ to ${\cal T}'$ creates three new  classes in $C_0^*$ ($C_0^*$ now being split into the class $C_0^* \cap e \cap e'$ which contains vertices from $C_i$, the $G$-portion
of $e_i$, the $G$-portion of $e_i'$ and the class of vertices not in any edge). 
Adding each other edge from $\mathcal{T}''$ to $\mathcal{T}'$ increases the  
number of classes in $C_0^*$ by one
(as by Claim~C we note that some
vertex in $G \cap C_0^*$ is not contained in any edge in $\mathcal{T}''$).
Also note that for every $s\in [t]$
some edge in $\mathcal{T}''$ cuts $C^*_s$.

 So let $\mathcal{T}''' = \mathcal{T}' \cup \mathcal{T}'' = (\mathcal{T}\setminus{e}) \cup \mathcal{T}''$.
Removing $e$ from ${\cal T}$ decreases the number of classes by $t+2$ and
adding $\mathcal{T}''$ increases the number of classes by at least $t+|\mathcal{T}''|+1$. So by increasing the number of edges by $|\mathcal{T}''|-1$ we have increased the number of classes
by at least $|\mathcal{T}''|-1$ and therefore we still have at least $k$ more classes than edges.
\end{proof}

\section{Discussion}\label{sec:d}

The two main results proved in this paper are the existence of
polynomial-size kernels for {\sc Test-$r$-Cover}$(m-k)$ and {\sc
Test-$r$-Cover}$(n-k)$. In fact, our result for {\sc TestCover}$(m-k)$ is
stronger: {\sc TestCover}$(m-k)$ has a polynomial-size kernel for the
parameter $k+r$. It would be interesting to find out whether {\sc
TestCover}$(n-k)$ has a polynomial-size kernel for the parameter $k+r$.
In addition, Theorem \ref{thm:m-kfpt} gives an algorithm for {\sc TestCover}$(m-k)$
with running time $(f(r))^{O(k)}(n+m)^{O(1)}$; it would be interesting to see
if an algorithm with similar running time can be found for 
{\sc TestCover}$(n-k)$.

\section{Acknowledgments}

We are grateful to Manu Basavaraju and Mathew Francis for carefully reading an earlier version of this paper and informing us of a subtle flaw in Theorem \ref{thm:thm8} which led us to changing the proof substantially.


\begin{thebibliography}{1}

\bibitem{AloGutKimSzeYeo11}
N.~Alon, G.~Gutin, E.~J. Kim, S.~Szeider, and A.~Yeo,
Solving {MAX}-$r$-{SAT} above a tight lower bound.
Algorithmica 61 (2011), 638--655.

\bibitem{BDFH09} H.L.~Bodlaender, R.G.~Downey, M.R.~Fellows, and D.~Hermelin, On problems without polynomial kernels,
{\em J. Comput. Syst. Sci.}, 75(8): 423--434, 2009.



\bibitem{Bon72} J.A. Bondy, Induced subsets. {\em J. Combin. Th., Ser. B} 12
(1972), 201--202.

\bibitem{CFGJRTY} R. Crowston, M. Fellows, G. Gutin, M. Jones, F. Rosamond,
S. Thomass{\'e}, and A. Yeo,
Simultaneously Satisfying Linear Equations Over $\mathbb{F}_2$: MaxLin2 and
Max-$r$-Lin2 Parameterized Above Average.
{\em Proc. FSTTCS 2011}, LIPICS Vol. 13, 229--240.

\bibitem{CrowstonGJSY12} R. Crowston, G. Gutin, M. Jones, S. Saurabh, and A. Yeo,
Parameterized Study of the Test Cover Problem. In {\em MFCS 2012}, Lect. Notes Comput. Sci. 7464 (2012), 283-295.

\bibitem{CroGutJonYeo12} R. Crowston, G. Gutin, M. Jones, and A. Yeo, A new
lower bound on the maximum number of
satisfied clauses in {Max-SAT} and its algorithmic applications.
Algorithmica 64 (2012), 56--68.

\bibitem{CroJonMni} R. Crowston, M. Jones, and M. Mnich,
Max-Cut Parameterized above the Edwards-Erd{\H o}s Bound, In {\em  ICALP
2012}, Lect. Notes Comput. Sci. 7391 (2012) 242--253.

\bibitem{Debnb} K.M.J. De Bontridder,  B.J. Lageweg, J.K. Lenstra, J.B. Orlin, and L. Stougie, Branch and
bound algorithms for the test-cover problem. In {\em ESA 2002},  Lect. Notes Comput. Sci. 2461 (2002), 223--233.

\bibitem{De} K.M.J. De Bontridder, B.V. Halld\'orsson, M.M. Halld\'orsson, C.A.J. Hurkens, J.K. Lenstra, R. Ravi, and L. Stougie,
Approximation algorithms for the test cover problem, Math. Programming-B 98 (1–3) (2003), 477--491.

\bibitem{DLS09} M.~Dom, D.~Lokshtanov, and S.~Saurabh, Incompressibility though Colors and IDs,
In {\em 36th ICALP, Part I}, Lect. Notes Comput. Sci. 5555: 378--389, 2009.

\bibitem{DowneyFellows99}
R.~G. Downey and M.R.~Fellows, Parameterized Complexity, Springer, 1999.

\bibitem{FahTie} T. Fahle and K. Tiemann,
A faster branch-and-bound algorithm for the test-cover problem based on set-covering techniques,
J. Exp. Algorithmics 11 (2007), 2.2.

\bibitem{FlumGrohe06}
J.~Flum and M.~Grohe, Parameterized Complexity Theory, Springer Verlag, 2006.

\bibitem{GareyJohnson} M.R. Garey and D.S. Johnson, Computers and Intractability. A guide to the theory of NP-completeness. CA, Freeman, 1979.

\bibitem{GutKimMniYeo} G. Gutin, E.J. Kim, M. Mnich, and A. Yeo,
Betweenness Parameterized Above Tight Lower
Bound. J. Comput. Syst. Sci. 76 (2010), 872--878.

\bibitem{GutinMY12}
G. Gutin, G. Muciaccia, and A. Yeo, (Non-)existence of Polynomial Kernels for the Test Cover Problem. 
Inform. Proc. Lett. 113 (2013), 123--126

\bibitem{HalHalRav01} B.V. Halld\'orsson, M.M. Halld\'orsson, and R. Ravi,
On the approximability of the Minimum Test Collection problem. In {\em ESA 2001},
Lect. Notes Comput. Sci. 2161 (2001), 158--169.

\bibitem{HalMinRav01} B.V. Halld\'orsson, J.S. Minden, and R. Ravi, PIER: Protein identification by epitope recognition.
In {\em Currents in Computational Molecular Biology 2001}, 109--110, 2001. 	
	
\bibitem{Karp72} R.M. Karp. Reducibility among combinatorial problems,
In R.E. Miller and J.W. Thatcher, editors, Complexity of Computer Computations, pp. 85–103. Plenum Press, 1972.

\bibitem{MR99} M.~Mahajan and V.~Raman, Parameterizing  above guaranteed
values: MaxSat and MaxCut. J. Algorithms 31(2) (1999), 335--354.

\bibitem{MahajanRamanSikdar09} M.~Mahajan, V.~Raman, and S.~Sikdar,
Parameterizing above or below guaranteed values.
 J. Comput. Sys. Sci. 75(2):137--153, 2009.

\bibitem{MorSha85} B.M.E. Moret and H.D. Shapiro, On minimizing a set of tests,
SIAM J. Scientific \& Statistical Comput. 6 (1985), 983--1003.

\bibitem{Niedermeier06}
R.~Niedermeier, Invitation to Fixed-Parameter Algorithms,
Oxford University Press, 2006.


\end{thebibliography}
\end{document}